\documentclass[10pt]{article}
\usepackage[english]{babel}
\usepackage[T1]{fontenc} 
\usepackage[utf8]{inputenc}

\usepackage{bbm,amsmath, amsfonts, amsthm, amssymb}
\usepackage{mathtools}
\usepackage{hyperref}
\usepackage{xcolor}
\usepackage{graphicx}
\usepackage{braket}
\usepackage{tikz-cd}
\usepackage{enumerate}
\usepackage[ruled, vlined]{algorithm2e}
\usepackage{listings}
\usepackage{xcolor}
\usepackage{tikz}
\usetikzlibrary{quantikz2}
\usepackage[textsize=tiny]{todonotes}
\usepackage{bbm}
\usepackage{bm}
\usepackage{shuffle}

\graphicspath{{Figures/}}

\definecolor{codegreen}{rgb}{0,0.6,0}
\definecolor{codegray}{rgb}{0.5,0.5,0.5}
\definecolor{codepurple}{rgb}{0.58,0,0.82}
\definecolor{backcolour}{rgb}{0.95,0.95,0.92}

\lstdefinestyle{mystyle}{
    backgroundcolor=\color{backcolour},   
    commentstyle=\color{codegreen},
    keywordstyle=\color{magenta},
    numberstyle=\tiny\color{codegray},
    stringstyle=\color{codepurple},
    basicstyle=\ttfamily\footnotesize,
    breakatwhitespace=false,         
    breaklines=true,                 
    captionpos=b,                    
    keepspaces=true,                 
    numbers=left,                    
    numbersep=5pt,                  
    showspaces=false,                
    showstringspaces=false,
    showtabs=false,                  
    tabsize=2
}
\numberwithin{equation}{section}
\newtheorem{theorem}{Theorem}[section]

\newtheorem{lemma}[theorem]{Lemma}
\newtheorem{proposition}[theorem]{Proposition}

\theoremstyle{definition}
\newtheorem{definition}[theorem]{Definition}
\newtheorem{remark}[theorem]{Remark}

\newtheorem{example}[theorem]{Example}

\usepackage{amsmath}

\newcommand{\bluebf}[1]{\textcolor{blue}{\ensuremath{\mathbf{#1}}}}

\newcommand{\eps}{\varepsilon}

\newcommand{\NN}{\mathbb{N}}
\newcommand{\RR}{\mathbb{R}}
\newcommand{\PP}{\mathbb{P}}

\newcommand{\D}{\mathrm{d}}

\newcommand{\Ff}{\mathcal{F}}

\newcommand{\EE}{\mathbb{E}}

\usepackage{mathrsfs}
\usepackage[margin=2cm]{geometry}

\usepackage[symbol]{footmisc} 
\usepackage[round]{natbib}

\title{Martingality in EFM Signature Volatility Models}

\author{
Dimitri Sotnikov$^{1,2}$\thanks{\texttt{dmitrii.sotnikov@polytechnique.edu}. The author gratefully acknowledges financial support from Engie Global Markets.}
\and
Adeline Viot$^{3}$\thanks{\texttt{adeline.viot@ens-paris-saclay.fr}\\ We would like to thank Eduardo Abi Jaber for enriching discussions and valuable feedback on the article.}
}

\date{\today}

\begin{document}

\maketitle

\begin{center}
$^1$CMAP, École Polytechnique\\
$^2$Engie Global Markets\\
$^3$Department of Mathematics, ENS Paris-Saclay
\end{center}

\begin{abstract}
We establish a martingality criterion for the Exponentially Fading Memory (EFM) signature volatility model, extending the recent result of \cite*{abijaber2025martingale}. In the proposed framework, the volatility process is given by a linear functional of the EFM signature of a time-augmented Brownian motion, which yields a globally well-behaved, fully non-Markovian volatility model admitting a stationary distribution. The stationary EFM signature volatility model is also covered by allowing for a random initial value of the signature. Moreover, the volatility process can be adjusted to the initial forward variance curve. For this model, excluding trivial cases, we prove that the price process is a true martingale if and only if the truncation order is odd and a sign condition involving the correlation parameter holds.

\end{abstract}

\medskip
\noindent\textbf{Keywords:} EFM signature, stochastic volatility, martingality, Ornstein--Uhlenbeck process.

\smallskip
\noindent\textbf{MSC 2020:} 60H10, 60G44, 91G30.

\section{Introduction}

Distinguishing between strict local martingales and true martingales is a theoretical question with important practical consequences for stochastic volatility models, in which the local martingale under consideration is the asset price. Indeed, this question must be addressed to justify the viability of a stochastic volatility model for pricing \citep{sepp2023robust_SV}; otherwise, one risks mispricing futures contracts, violating put--call parity and observing unexpected asymptotic behaviour of implied volatility smiles \citep*{cox2005local, heston2007options, jacquier2015impliedvolatilitystrictlocal}. 

Since the classical criterion of \citet{Novikov1973} is rarely useful for establishing true martingality in stochastic volatility models, different techniques have been developed, first in Markovian frameworks \citep{Lions2007, Mijatovic2012, bernard2017martingale}, where the volatility dynamics follow an SDE, such as the Hull--White model \citep{Sin98}, the SABR model \citep{Jourdain2004} or the Heston model \citep{Andersen2007}. These techniques were subsequently extended to non-Markovian models, including the rough Bergomi model \citep{Gassiat2018OnTM}, the Volterra Heston model \citep*{jaber2019affinevolterraprocesses}, the path-dependent volatility model \citep{Nutz2024}, the signature volatility model \citep*{abijaber2025martingale}, as well as generic non-Markovian frameworks \citep{Ruf_2015, Criens2018, Criens2020NoAI}.

We are specifically interested in signature volatility models, pioneered by \citet*{arribas2020sig} and subsequently studied by \citet*{cuchiero2022theocalib, cuchiero2023spvix} and \citet{abijaber2024signature}. In this class of models, the volatility process is represented as a linear function of the signature of the underlying stochastic factor, the linear form being justified by the universal approximation property of signatures or by the possibility of representing the volatility process exactly as such a linear function. In practice, only finite linear combinations of signature elements (up to order $N$) can be considered, and their coefficients must be chosen so as to ensure the martingality of the price process. For the signature of the time-augmented Brownian motion, and in the presence of a nonzero dominant coefficient in front of $W_t^N / N!$, \citet*{abijaber2025martingale} showed that the price process is a true martingale if and only if $N$ is odd and the product of the dominant coefficient with the correlation parameter is nonpositive.

In this article, we aim to extend this result in three directions. The first, and most important, concerns the stationarity (or the existence of a stationary limit) of the volatility process. Indeed, finite linear functions of the classical signature do not, in general, yield a nontrivial stationary process: even when the signature is computed for a stationary process, as proposed by \citet*{cuchiero2023spvix}, the resulting volatility does not admit a stationary limit and exhibits highly unrealistic long-term behaviour. This was observed in \cite{abijaberEFM} and led to the introduction of a ``stationarised'' signature, called the Exponentially Fading Memory (EFM) signature. We therefore consider the EFM signature volatility model, in which the volatility is given by a linear combination of the elements of the EFM signature of the Brownian motion.

The second extension consists in allowing for a random initial volatility $\Sigma_0$ by taking a random initial value of the signature. First, this allows the volatility process to be initialised with its stationary distribution, thereby ensuring its global stationarity. Second, a random initial volatility may be beneficial for calibration, as observed in the Markovian case by \citet*{mechkov2015hot, Jacquier2019, lemaire2022stationary}.

Finally, we allow the model to take into account the initial forward variance curve, as is done in the Bergomi model and in the more recent Quintic model of \citet*{quintic2022}, which is a particular case of the proposed framework.

The model thus takes the following form: given two Brownian motions $(W, B)$ with correlation parameter $\rho \in [-1, 1]$ and a \( d \)-dimensional Brownian motion \( Z \) independent of \( (W, B) \), define
\begin{equation}\label{syst intro}
dS_t=S_t\Sigma_t\,dB_t,\qquad \Sigma_t=g(t)\braket{\bm{\sigma},{\bf D}^{\lambda}_t\bm\Xi \otimes\widehat{\mathbb{Y}}^{\lambda}_{0,t}}, \quad t \in [0, T],
\end{equation}
where $\widehat{\mathbb{Y}}^{\lambda}_{0,t}$ denotes the EFM signature of the time-augmented Brownian motion $\widehat{Y}_t = (t, W_{t}, Z_{t}^1, \ldots, Z_{t}^d)$ over $[0, t]$ with decay parameter $\lambda$ (see Definition~\ref{def:efm-def} below), $\bm\Xi$ stands for the random initial value of the EFM signature, ${\bf D}^{\lambda}_t$ is a discounting operator arising from the algebraic structure of the EFM signature, $\bm\sigma \in T^N(\RR^{2+d})$ specifies the coefficients of the linear function of degree $N \in \NN$, and $g$ is a positive function such that $0 < g_* \leq g(t) \leq g^*$ for all $t \in [0, T]$.

In our main result, we prove that the martingality criterion identified in \cite*[Theorems~3.1 and 3.5]{abijaber2025martingale} also holds for the EFM signature volatility model introduced above: namely, in the nontrivial case $N \geq 2$ and $\rho \neq 0$,
\begin{center}
    \textit{$S$ is a true martingale if and only if $N$ is odd and $\rho \bm{\sigma}^{\bluebf{2}^{\otimes N}} < 0$.}
\end{center}

An odd truncation order of the linear combination, which plays a role analogous to the degree of a polynomial, is a characteristic feature of this class of models. For polynomial volatility models, a condition of this type already appears in \cite*[Proposition~A.1]{quintic2022}, where an odd degree, together with a sign condition on the correlation, is shown to be sufficient for the martingality of the Quintic model.
 Our result extends this in two directions: it provides a full characterisation for polynomials of arbitrary degree, and it generalises the result to the non-Markovian setting of EFM signatures.

The proof relies first on the classical change-of-measure argument of \citet{Sin98}, which reduces the question of true martingality to the study of the explosion of an SDE involving the EFM signature. However, the proofs of \cite*{abijaber2025martingale} do not apply directly in our case, since the algebraic structure of the EFM signature differs from that of the classical signature. We therefore introduce an auxiliary signature-SDE, which is more natural in the EFM context and has the same explosion time as the signature-SDE of \cite*[Theorem~4.2]{abijaber2025martingale}. We then adapt the signature bounds of \cite*[Lemmas~5.1--5.2]{abijaber2025martingale} to the EFM setting in order to prove that the solution to this auxiliary SDE does not explode when $N$ is odd and $\rho \bm{\sigma}^{\bluebf{2}^{\otimes N}} < 0$, and that it explodes on $[0, T]$ with positive probability otherwise.

The paper is organised as follows. Section~\ref{sec : classical model} recalls the definition of the EFM signature and the necessary preliminaries. Section~\ref{subsec : model efm} presents the associated stochastic volatility model, states the martingality criterion and shows how the problem reduces to the study of the explosion of an EFM-sig-SDE. Section~\ref{sec : efm proofs} is devoted to the proof of the main result. 

\section{Preliminaries on the EFM Signature}\label{sec : classical model}

This section provides the necessary preliminaries and recalls the notion of the EFM signature, together with the properties needed for our analysis.

\subsection{Tensor Algebra}\label{subsec : tensor}
We fix $d\in \NN$. For any $n\in\NN^*$, $(\RR^{{ d}})^{\otimes n}$ denotes the space of tensors of order $n$, and $(\RR^d)^{\otimes 0}:=\RR$. We then define the extended tensor algebra over $\RR^d$ as
    $$T((\RR^d)):=\{\bm{\ell}=(\bm{\ell}^n)^{\infty}_{n=0} : \bm{\ell}^n\in(\RR^d)^{\otimes n} \text{ for all } n\in\NN \}.$$
The tensor algebra over $\RR^d$ is defined as the set of tensor sequences for which there exists an order $N$ such that every component of order strictly greater than $N$ is zero. We denote the tensor algebra by $T(\RR^d)=\bigcup_{N\in\NN} T^N(\RR^d)$, where
    $T^N(\RR^d):=\{\bm{\ell}\in T((\RR^d)) : \bm{\ell}^n=0 \text{ for all } n > N \}.$

Let $\{e_1,\dots,e_d\}$ denote the canonical basis of $\RR^d$. Then, for all $n\in\NN^*$, $\{e_{i_1}\otimes \dots \otimes e_{i_n}\}_{i_1,\dots,i_n\in\{1,\dots,d\}}$ forms a basis of $(\RR^d)^{\otimes n}$. We identify the vector $e_{i_1}\otimes \dots \otimes e_{i_n}$ with the word \bluebf{i_1\dots i_n} over the alphabet $A_d=\{\bluebf{1},\dots,\bluebf{d}\}$. 

The set of words of length $n$ is denoted by
    $V_n:=\{\bluebf{i_1\dots i_n}: \bluebf{i_k}\in A_d \text{ for all } k\in\{1,\dots,n\}\}.$
We denote by \bluebf{\varnothing} the empty word and set $V_0=\{\bluebf{\varnothing}\}$. We identify the empty word with the basis element of $\RR$. Thus, $V=\bigcup_{n\geq 0} V_n$ can be identified with a basis of $T((\RR^d))$, and any $\bm{\ell}\in T((\RR^d))$ can be decomposed as
    $$\bm{\ell}=\sum_{n=0}^\infty \sum_{\bluebf{v}\in V_n} \bm{\ell}^{\bluebf{v}} \bluebf{v}.$$

We denote by $\bm{\ell}\bluebf{v}$ the tensor sequence $\bm{\ell}\otimes e_{i_1}\otimes \dots \otimes e_{i_n}$, where $\bm{\ell}\in T((\RR^d))$ and $\bluebf{v}=\bluebf{i_1\dots i_n}$. 

For all $\bm{p}, \bm{q}\in T((\RR^d))$, we define their tensor product by
    $$\bm{p}\otimes\bm{q}=(\bm r^0,\bm r^1,\dots),\quad \bm r^n=\sum_{k=0}^n \bm p^k\otimes \bm q^{n-k}, \quad n\in \NN.$$

We define the bracket between $\bm{\ell}\in T(\RR^d)$ and $\bm{p}\in T((\RR^d))$ by
    $$\braket{\bm{\ell},\bm{p}}:=\sum_{n\geq 0}\sum_{\bluebf{v}\in V_n} \bm{\ell}^{\bluebf{v}} \bm{p}^{\bluebf{v}}.$$

\begin{definition}[\textit{Shuffle product}]
For two words $\bluebf{v}, \bluebf{w}$ and two letters $\bluebf{i}, \bluebf{j}\in A_d$, the \textit{shuffle product} is the map $\shuffle : V\times V \to T(\RR^d)$ defined recursively by
\begin{align*}
    (\bluebf{vi})\shuffle (\bluebf{wj})&=(\bluebf{v}\shuffle (\bluebf{wj}))\bluebf{i}+((\bluebf{vi})\shuffle \bluebf{w})\bluebf{j},\\
    \bluebf{v}\shuffle \bluebf{\varnothing}&=\bluebf{\varnothing}\shuffle\bluebf{v}=\bluebf{v}.
\end{align*}
The shuffle product is then extended to $T(\RR^d)$ by bilinearity. 
\end{definition}

\subsection{EFM Signatures}\label{subsec : efm signa}

In this subsection, we recall the definition of the exponentially fading memory (EFM) signature, introduced in~\cite{abijaberEFM}.

For two continuous semimartingales $X$ and $Y$, we denote the Itô integral by $\int_0^\cdot Y_t dX_t$ and the Stratonovich integral by $\int_0^\cdot Y_t \circ dX_t = \int_0^\cdot Y_t dX_t  + \frac{1}{2} [X,Y]_\cdot.$
We fix a time horizon $T>0$ and work on a filtered probability space $(\Omega,\Ff,(\Ff_t)_{t\in [0, T]},\PP)$ on which a continuous $\RR^d$-valued semimartingale $(X_t)_{t\in [0, T]}$ is defined. We also fix $\lambda =(\lambda_i)_{i=1, \ldots, d}\in\RR^d$ such that $\lambda_i \geq 0$ for all $i =1, \ldots, d$. We set $\Delta:=\{(s,t)\in [0, T]^2 : s\leq t\}$.

\begin{definition}\label{def:efm-def}
The $\lambda$-EFM signature of $X$ is defined by 
$$\mathbb{X}^\lambda\colon(s,t,\omega)\in \Delta\times\Omega\mapsto\mathbb{X}^\lambda_{s,t}=(1,\mathbb{X}_{s,t}^{\lambda,1}(\omega),\mathbb{X}_{s,t}^{\lambda,2}(\omega),\dots)\in  T((\RR^d)),$$
where, for any $\bluebf{v}=\bluebf{i_1\dots i_n}\in V$,
$$\mathbb{X}_{s,t}^{\lambda,\bluebf{i_1\dots i_n}}:=\int_s^t e^{-\lambda_{\bluebf{i_1\dots i_n}}(t-u)}\mathbb{X}_{s,u}^{\lambda,\bluebf{i_1\dots i_{n-1}}}\circ dX_u^{i_n},\qquad \lambda_{\bluebf{i_1\dots i_n}}:=\sum_{k=1}^n\lambda_{i_k},$$
with $\lambda_{\bluebf{\varnothing}}=0$ and $\lambda_{\bluebf{uv}}=\lambda_{\bluebf{u}}+\lambda_{\bluebf{v}}$ for all $\bluebf{u},\bluebf{v}\in V$.
\end{definition}

For \( s \in \mathbb{R} \), let ${\bf\Lambda},{\bf D}^{\lambda}_s:T((\mathbb{R}^d))\to T((\mathbb{R}^d))$ be the operators defined, for all $\bm\ell=\sum_{n=0}^\infty\sum_{\bluebf{v}\in V_n}\bm\ell^{\bluebf{v}}\bluebf{v}\in T((\mathbb{R}^d))$, by
    $$
    {\bf\Lambda}\bm\ell=\sum_{n=0}^\infty\sum_{\bluebf{v}\in V_n}\lambda_{\bluebf{v}}\bm\ell^{\bluebf{v}}\bluebf{v},
    \qquad
    {\bf D}^{\lambda}_s\bm\ell=\sum_{n=0}^\infty\sum_{\bluebf{v}\in V_n}e^{-\lambda_{\bluebf{v}}s}\bm\ell^{\bluebf{v}}\bluebf{v}.
    $$
Then, according to~\cite[Proposition 3.6]{abijaberEFM}, the definition of the EFM signature can be summarised by the following stochastic differential equation:
    $$
    d\mathbb{X}_{s,t}^{\lambda}
    =-{\bf\Lambda}\mathbb{X}_{s,t}^{\lambda}\,dt
    +\mathbb{X}_{s,t}^{\lambda}\otimes\circ dX_t,
    \quad t \in [s, T].
    $$

The following result, taken from~\cite[Proposition 4.1]{abijaberEFM}, is the EFM counterpart of Chen's relation for classical signatures; see~\cite[Theorem 1.22]{chevyrev2022primer}.

\begin{proposition}\label{prop : chen EFM}
    Let $X=(X_t)_{t\in[0,T]}$ be a continuous $\RR^d$-valued semimartingale. Then, 
    \begin{equation}\label{eq:Chen_EFM}
        \mathbb{X}^{\lambda}_{s,t}=({\bf D}^{\lambda}_{t-u}\mathbb{X}^{\lambda}_{s,u})\otimes \mathbb{X}^{\lambda}_{u,t}, \quad 0\leq s\leq u\leq t\leq T.
    \end{equation}
\end{proposition}

We recall the fundamental linearisation property linking the {shuffle product} to path signatures; see~\cite[Proposition 2.2]{gaines1994algebra} and~\cite[Proposition 4.8]{abijaberEFM}.

\begin{proposition}\label{prop shuffle signa}
For all $\bm{\ell_1},\bm{\ell_2}\in T(\RR^{d})$, we have
    $$\braket{\bm{\ell_1},\mathbb{X}^{\lambda}_{s,t}}\braket{\bm{\ell_2},\mathbb{X}^{\lambda}_{s,t}} = \braket{\bm{\ell_1\shuffle \bm{\ell_2}},\mathbb{X}^{\lambda}_{s,t}}.$$
\end{proposition}
An element $\mathbb{X}\in T((\RR^d))$ with \( \mathbb{X}^{\bluebf{\varnothing}} = 1 \) is called {group-like}
if $\braket{\bm{\ell_1},\mathbb{X}}\braket{\bm{\ell_2},\mathbb{X}}=\braket{\bm{\ell_1}\shuffle\bm{\ell_2},\mathbb{X}}$ for all $\bm{\ell_1},\bm{\ell_2}\in T(\RR^d)$. Proposition~\ref{prop shuffle signa} thus states that EFM signatures are group-like.

\section{EFM Signature Volatility Model and Main Result}\label{subsec : model efm}

Fix \( T > 0 \) and let $(\Omega,\Ff,(\Ff_t)_{t\in[0,T]},\PP)$ be a filtered probability space supporting two one-dimensional Brownian motions $B$ and $W$ with correlation $\rho\in[-1,1]$, and a $d$-dimensional Brownian motion $Z$ independent of $(B, W)$. We assume that, under the risk-neutral measure $\PP$, the price process $S$ is given by the stochastic volatility model
    \begin{equation}\label{syst1 signa EFM}
    dS_t=S_t\Sigma_t\,dB_t,\qquad \Sigma_t=g(t)\braket{\bm{\sigma}, {\bf D}^{\lambda}_t\bm\Xi \otimes \widehat{\mathbb{Y}}^{\lambda}_{0,t}}, \qquad t \in [0, T],
    \end{equation}
where
$\widehat{\mathbb{Y}}^{\lambda}_{0,t}$ is the EFM signature of the process $\widehat{Y}\colon t \mapsto (t,W_t, Z_{t}^1, \ldots, Z_{t}^d)$ with parameter \( \lambda \in \mathbb{R}^{2+d}_{+} \) over $[0,t]$, $\bm{\sigma} \in T^N(\mathbb{R}^{2+d})$ for some \( N \in \mathbb{N} \), $\bm\Xi$ is an $\mathcal{F}_0$-measurable
$T((\mathbb{R}^{2+d}))$-valued group-like random variable such that $\bm\Xi ^{\bluebf{v}}$ has finite moments of every order for all \( \bluebf{v} \in V \), and $g$ is a positive function such that \( g_{*} \leq g(t) \leq g^*\) for all \( t \in [0, T] \), for some constants \( g_{*}, g^* > 0 \).

We will also consider the model with \( d = 0 \), in which case the process \( \widehat{Y} \) is simply given by the time-augmented Brownian motion \( \widehat{Y} = \widehat{W}\colon t \mapsto (t, W_{t}). \)

\begin{remark}\label{rmk:Xi_rmk}
    In view of Chen's identity~\eqref{eq:Chen_EFM}, the random variable \( \bm\Xi \) has a clear interpretation as the initial condition at \( t = 0 \) of the process \( t \mapsto  \widehat{\mathbb{Y}}_{-\infty, t}^\lambda = {\bf D}_{t}^\lambda \widehat{\mathbb{Y}}_{-\infty, 0}^\lambda \otimes \widehat{\mathbb{Y}}_{0, t}^\lambda \), which summarises the behaviour of the stochastic factors driving the volatility before the initial time. This is a very natural assumption that allows the model to be initialised with a ``non-empty'' memory, a possibility absent from standard signature volatility models. Indeed, unlike the classical signature, the EFM signature with \( \lambda_{i} > 0 \) for all \( i = 1, \ldots, 2+d \) can be defined, thanks to the exponential damping, on infinite intervals such as \( (-\infty, 0] \), which makes it possible to include this initial information. We do not discuss the rigorous construction here and refer the interested reader to~\cite[Section~3]{abijaberEFM}.
\end{remark}

The unique strong solution to~\eqref{syst1 signa EFM} is given by the exponential local martingale
$$S_t=\exp\left(\int_0^t \Sigma_s\,dB_s-\frac12\int_0^t \Sigma_s^2\,ds\right),
\qquad t\in[0,T],$$
where both integrals are well-defined since \( \mathbb{E}\left[\int_{0}^T \Sigma_{t}^2\,dt\right] < \infty\). This follows from the finiteness of the expected EFM signature~\cite[Theorem~5.5]{abijaberEFM} and the integrability conditions imposed on \( \bm\Xi \).

We now present three stochastic volatility models that are particular cases of \eqref{syst1 signa EFM}.

\begin{example}\label{example : volatility}
    \begin{enumerate}[(a)]
        \item The (classical) \textit{signature volatility model} of \citet{abijaber2024signature}  is recovered by choosing $\lambda = 0$, $\bm\Xi\equiv \bluebf{\varnothing}$, \( d = 0 \) and \( g \equiv 1 \). The multidimensional version of this model studied in~\cite{abijaber2025martingale} is obtained by taking \( d > 0 \) instead of \( d = 0 \) and, in addition to the previous assumptions, setting \( \rho = 1 \).  In both cases, the martingality criterion was established in~\cite[Theorems~3.1 and 3.5]{abijaber2025martingale}.
        
        \item Consider model \eqref{syst1 signa EFM} with \( \bm\sigma = \sum_{k = 0}^N \alpha_{k} \bluebf{2}^{\otimes k}\), \( \lambda_2 > 0 \), \( \bm\Xi = \bluebf{\varnothing} \), and $g(t) = \dfrac{\sqrt{\xi_{0}(t)}}{\sqrt{\mathbb{E}[\braket{\bm{\sigma},\widehat{\mathbb{Y}}^{\lambda}_{0,t}}^2]}}$, where \( \xi_{0} \) denotes the initial forward variance curve. Setting $N = 5$ and $\alpha_{2} = \alpha_{4} = 0$, we recover the \textit{Quintic model} of \citet{quintic2022} up to a trivial reparametrisation. \cite[Proposition~A.1]{quintic2022} establishes the martingale property in the case \( \alpha_{5} > 0 \) and \( \rho \leq 0 \).
        
        \item As mentioned in Remark~\ref{rmk:Xi_rmk}, the EFM signature \( \widehat{\mathbb{Y}}_{-\infty, t}^\lambda \) of the time-augmented two-sided Brownian motion $(\widehat{Y}_t)_{t \in \mathbb{R}}$ can be defined over the interval \( (-\infty, t] \) when \( \lambda_{i} > 0 \) for all \( i = 1, \ldots, 2+d \). In this case, the process \( \langle \bm \sigma,  \widehat{\mathbb{Y}}_{-\infty, t}^\lambda \rangle \) is stationary and ergodic; see~\cite[Theorems 3.17 and 5.3]{abijaberEFM}. Hence, the \textit{stationary EFM signature volatility model} is obtained by taking the stationary initialisation $\bm\Xi = \widehat{\mathbb{Y}}_{-\infty, 0}^\lambda$. The integrability assumption on \( \bm\Xi \) is satisfied since the expected signature $\mathbb{E}[\widehat{\mathbb{Y}}_{-\infty, 0}^\lambda]$ is finite; see~\cite[Theorem~5.5]{abijaberEFM}.
    \end{enumerate}
\end{example}

Our main result provides conditions under which the price process \( S \) defined by~\eqref{syst1 signa EFM} is a true martingale, under the assumption that the leading coefficient $\bm{\sigma}^{\bluebf{2}^{\otimes N}}$ in front of \( W_{t}^N / N! \) does not vanish. It extends the results of~\cite[Theorems~3.1 and 3.5]{abijaber2025martingale} and \cite[Proposition~A.1]{quintic2022} to the more general model considered here.

\begin{theorem}\label{th : crit EFM 1d}
    Let $S$ be given by~\eqref{syst1 signa EFM}. Suppose that the leading coefficient $\bm{\sigma}^{\bluebf{2}^{\otimes N}}$ is nonzero.
    \begin{enumerate}[(i)]
        \item If $N=1$ or $\rho=0$, then $S$ is a martingale.
        \item If $N\geq 2$ and $\rho\neq 0$, then $S$ is a martingale if and only if $N$ is odd and $\rho\bm{\sigma}^{\bluebf{2}^{\otimes N}}<0$.
    \end{enumerate}
\end{theorem}

\subsection{Martingality and Explosion of an EFM-Sig-SDE}

In this subsection, we show that the martingality of $S$ reduces to the study of the explosion of the solution to
\begin{equation}\label{eds X}
X_t= \rho\int_0^tg(s)\braket{\bm{\sigma},{\bf D}^{\lambda}_s\bm\Xi\otimes\widehat{\mathbb{Y}}^{X,\lambda}_{0,s}}\,ds+W_t,\qquad t\in[0,T],
\end{equation}
where $\widehat{\mathbb{Y}}^{X,\lambda}_{0,t}$ denotes the \textit{EFM} signature of $t \mapsto (t,X_t,Z_t^1,\dots,Z_t^d)$ over $[0,t]$.

The following lemma allows us to rewrite the drift of \( X \) as a linear functional of \( \widehat{\mathbb{Y}}^{X,\lambda}_{0,t} \).

\begin{lemma}\label{lem:projection}
    There exists an \( \mathcal{F}_0 \)-measurable \( T^N(\mathbb{R}^{2+d}) \)-valued process \( \bm\ell = (\bm\ell_{t})_{t \in [0, T]} \) such that 
    \[
        g(t)\braket{\bm{\sigma},{\bf D}^{\lambda}_t\bm\Xi\otimes\widehat{\mathbb{Y}}^{X,\lambda}_{0,t}} = \langle \bm\ell_{t}, \widehat{\mathbb{Y}}^{X,\lambda}_{0,t} \rangle, \quad t \in [0, T],
    \]
    where 
    \begin{equation}\label{eq:ell_def}
        \bm\ell_{t}^{\bluebf{v}} = \sum_{\bluebf{u} \in V}g(t)\bm\sigma^{\bluebf{uv}}e^{-\lambda_{\bluebf{u}}t}\bm\Xi^{\bluebf{u}}, \quad \bluebf{v} \in V.
    \end{equation}
    In particular, its leading term is given by \( \bm{\ell}^{\bluebf{2}^{\otimes N}}_{t} = g(t)\bm{\sigma}^{\bluebf{2}^{\otimes N}}\).
\end{lemma}
\begin{proof}
    By definition of the tensor product, we have
    \begin{align*}
        g(t)\braket{\bm{\sigma},{\bf D}^{\lambda}_t\bm\Xi\otimes\widehat{\mathbb{Y}}^{X,\lambda}_{0,t}} &= g(t)\sum_{n \geq 0}\sum_{|\bluebf{v}| = n}\bm\sigma^{\bluebf{v}}\bigg(\sum_{\substack{\bluebf{u}, \bluebf{w} \in V, \\ \bluebf{uw} = \bluebf{v}}}e^{-\lambda_{\bluebf{u}}t}\bm\Xi^{\bluebf{u}}\widehat{\mathbb{Y}}^{X,\lambda, \bluebf{w}}_{0,t}\bigg) \\ 
        &= \sum_{n \geq 0}\sum_{|\bluebf{w}| = n}\bigg(\sum_{\bluebf{u} \in V}g(t)\bm\sigma^{\bluebf{uw}}e^{-\lambda_{\bluebf{u}}t}\bm\Xi^{\bluebf{u}}\bigg)\widehat{\mathbb{Y}}^{X,\lambda, \bluebf{w}}_{0,t},
    \end{align*}
    so that \eqref{eq:ell_def} holds, and the sum is finite since \( \bm\sigma \in T^N(\mathbb{R}^{2+d}) \).
\end{proof}

We first establish the existence of a solution to~\eqref{eds X} up to the explosion time $\tau_\infty^X:=\lim_{n\to\infty}\tau_n^X$, where $\tau_n^X:=\inf\{t\in[0,T]:|X_t|=n\}$ with the convention $\inf\varnothing=+\infty$. However, rather than studying the explosion time \( \tau_{\infty}^X \) as in \cite[Sections~6--7]{abijaber2025martingale}, it is more convenient in the EFM signature setting to work with the process $U_t:=\langle \bluebf{2}, \widehat{\mathbb{Y}}^{X,\lambda}_{0,t}\rangle$ and to analyse its explosion time $\tau_\infty^U$ (defined analogously). Whenever $X$ is well-defined, $U$ satisfies
\begin{equation}\label{eds U}
 dU_t=(-\lambda_2U_t+\rho \braket{\bm{\ell}_{t}, \widehat{\mathbb{Y}}^{X,\lambda}_{0,t}})\,dt+dW_t,\qquad U_0=0.
\end{equation}
The proposition below shows that these two explosion times are in fact equal.

\begin{proposition}\label{prop : well def of X}
    The SDEs~\eqref{eds X} and~\eqref{eds U} each admit a unique strong solution. Moreover,
    \begin{equation}\label{lien temps X}
        \tau^X_\infty=\tau^U_\infty=\inf\left\{t\in[0,T] : \int_0^t \braket{\bm{\ell}_{s}, \widehat{\mathbb{Y}}^{X,\lambda}_{0,s}}^2\,ds = \infty\right\}.
    \end{equation}
\end{proposition}

\begin{proof}
    The case \( \rho = 0 \) is trivial, with \( \tau_{\infty}^X = \tau_{\infty}^U = \infty \). We now assume that \( \rho \neq 0 \).

    We first establish strong existence and uniqueness for~\eqref{eds U}.
    This equation is equivalent to a system for $U=\widetilde{\mathbb{Y}}^{\bluebf{2}}$, where $(\widetilde{\mathbb{Y}}^{\bluebf{w}})_{|\bluebf{w}|\leq N}$ satisfies $\widetilde{\mathbb{Y}}^{\bluebf{\varnothing}}\equiv 1$ and, for all $|\bluebf{w}|<N$,
    \begin{align*}
        d\widetilde{\mathbb{Y}}_t^{\bluebf{w1}}&=(-\lambda_{\bluebf{w1}}\widetilde{\mathbb{Y}}_t^{\bluebf{w1}} +\widetilde{\mathbb{Y}}_t^{\bluebf{w}})\,dt,\\
        d\widetilde{\mathbb{Y}}_t^{\bluebf{w2}}&=(-\lambda_{\bluebf{w2}}\widetilde{\mathbb{Y}}_t^{\bluebf{w2}} +\widetilde{\mathbb{Y}}_t^{\bluebf{w}}\rho\sum_{|\bluebf{u}|\leq N }\bm{\ell}_{ t}^{\bluebf{u}}\widetilde{\mathbb{Y}}_t^{\bluebf{u}})\,dt + \widetilde{\mathbb{Y}}_t^{\bluebf{w}} \circ dW_t, \\
        d\widetilde{\mathbb{Y}}_t^{\bluebf{wi}}&=-\lambda_{\bluebf{wi}}\widetilde{\mathbb{Y}}_t^{\bluebf{wi}}\,dt +\widetilde{\mathbb{Y}}_t^{\bluebf{w}} \circ dZ^{i-2}_t, \quad i\in\{3,\dots,2+d\},
    \end{align*}
    with the initial condition \( \widetilde{\mathbb{Y}}^{\bluebf{w}}_{0} = 0 \) for all $\bluebf{w}$ such that $0 < |\bluebf{w}| \leq N$.
    The coefficients of the system are polynomial in the state variables $(\widetilde{\mathbb{Y}}^{\bluebf{w}})_{|\bluebf{w}|\leq N}$, with $\mathcal F_0$-measurable polynomial coefficients that are bounded in \( t \) for each realisation of \( \bm\Xi \); they are therefore locally Lipschitz in the state variables. The standard existence and pathwise uniqueness theorem for Stratonovich SDEs extends to this setting and thus ensures the existence of a unique strong solution up to a possible explosion time.
    
    Define $X_t=\lambda_2\int_0^t U_s ds + U_t$. Then $dX_t=\rho\braket{\bm{\ell}_t, \widetilde{\mathbb{Y}}^{}_{t}}\,dt + dW_t$. One checks by induction that, for all $\bluebf{w}$ such that $|\bluebf w|\le N$, $\widetilde{\mathbb{Y}}^{\bluebf{w}}_t=\widehat{\mathbb{Y}}^{X,\lambda,\bluebf{w}}_{0,t}$,
   where the right-hand side denotes the corresponding component of the EFM signature of the process $(t,X_t,Z_t^1,\dots,Z_t^d)$ over $[0,t]$. This proves existence. Uniqueness follows from the uniqueness of the solution to the above system.
    
    We now study the explosion times. By integration by parts,
        $$X_t=U_t+\lambda_2\int_0^t U_s ds, \quad U_t=X_t-\lambda_2\int_0^t e^{-\lambda_2(t-s)}X_s ds,$$
    which implies $\tau^X_\infty=\tau^U_\infty$.
    
    On $\{\tau^X_\infty<\infty\}$, the Cauchy--Schwarz inequality yields
        $\rho^2\int_0^{\tau^X_\infty}\braket{\bm{\ell}_{s}, \widehat{\mathbb{Y}}^{X,\lambda}_{0,s}}^2 ds=\infty.$

    By Lemma~\ref{lemme espérance majo EFM} below, for all words $\bluebf w$ such that $|\bluebf w|\leq N$, we have
        $$\EE\left[\rho^2\int_0^{\tau_n^U\wedge t}g(s)^2\braket{\bluebf w,{\bf D}_s^{\lambda}\bm\Xi\otimes\widehat{\mathbb Y}^{X,\lambda}_{0,s}}^2ds\right]
        \leq
        (\rho g^*)^2C_{\lambda,\bluebf w,t}(\bm\Xi)n^{\frac{4|\bluebf w|+1}{2}}.
        $$
    Since $\bm{\sigma}\in T^N(\RR^{2+d})$, we deduce that, for all \( n \in \mathbb{N} \), $\rho^2\int_0^{\tau^U_n \wedge t} g(s)^2\braket{\bm{\sigma},{\bf D}^{\lambda}_s\bm\Xi \otimes\widehat{\mathbb{Y}}^{X,\lambda}_{0,s}}^2 ds <\infty$ almost surely, which yields the result.
\end{proof}

We now state the lemma used in the above proof, which is the analogue of~\cite[Lemma 5.2]{abijaber2025martingale} for EFM signatures. The proof of the first assertion closely follows that of~\cite[Lemma 5.2]{abijaber2025martingale} and is therefore omitted.

\begin{lemma}\label{lemme espérance majo EFM}
    Let $\lambda\in\mathbb{R}^{2+d}_{+}$, let $Z$ be a $d$-dimensional Brownian motion, and let $X:[0,T]\to\mathbb{R}$ be an Itô process such that $[Z^i,X]=0$ for all $i$. Denote by $\widehat{\mathbb{Y}}^{X,\lambda}_{0,t}$ the \textit{EFM} signature of $t \mapsto (t,X_t,Z_t^1,\ldots,Z_t^d)$. Then, for every $\bluebf{w}\in V$, $q\ge2$, $\varepsilon,\alpha>0$, and every stopping time $\tau\le T$,
        $$\mathbb{E}\!\left[\int_0^\tau|\langle\bluebf w,\widehat{\mathbb Y}^{X,\lambda}_{0,t}\rangle|^qdt \right]\le C_{\lambda,\bluebf w,q,\varepsilon,T}\mathbb{E}\!\left[\int_0^\tau\left(1+|\langle\bluebf2,\widehat{\mathbb Y}^{X,\lambda}_{0,t}\rangle|^{q(N_2(\bluebf w)+\varepsilon)}\right)dt \right].$$
    Moreover, for every $\mathcal F_0$-measurable $\bm\Xi$, 
        $$\mathbb{E}\!\left[\int_0^\tau|\langle\bluebf w,{\bf D}_t^{\lambda}\bm\Xi\otimes\widehat{\mathbb Y}^{X,\lambda}_{0,t}\rangle|^qdt\right]\le C_{\lambda,\bluebf w,q,\varepsilon,\alpha,T}(\bm\Xi)\left(1+\mathbb{E}\!\left[\int_0^\tau\left(1+|\langle\bluebf2,\widehat{\mathbb Y}^{X,\lambda}_{0,t}\rangle|^{q(1+\alpha)(N_2(\bluebf w)+\varepsilon)}\right)dt\right]\right),$$
    where $C_{\lambda,\bluebf w,q,\varepsilon,\alpha,T}(\bm\Xi)$ depends on $\bm\Xi$ only through the moments of $\bm\Xi^{\bluebf v}$, $\bluebf v\le\bluebf w$, and where $\bluebf v\le\bluebf w$ means that $\bluebf w=\bluebf{v\tilde v}$ for some $\bluebf{\tilde v}\in V$.

\end{lemma}

\begin{proof}
    Note that, for all $t\in[0,T]$, we have
        $\braket{\bluebf{w},{\bf D}^{\lambda}_t\bm\Xi\otimes\widehat{\mathbb{Y}}^{X,\lambda}_{0,t}}=\sum_{\substack{\bluebf{v},\bluebf{\tilde{v}}\\\bluebf{v\tilde{v}}=\bluebf{w}}}e^{-\lambda_{\bluebf{v}}t}\bm\Xi^{\bluebf{v}}\widehat{\mathbb{Y}}_{0,t}^{X,\lambda,\bluebf{{\tilde{v}}}}.$
    Then, by Jensen's inequality, it suffices to bound
    $
    \EE\left[\int_0^\tau |\bm\Xi^{\bluebf v}\widehat{\mathbb Y}_{0,t}^{X,\lambda,\bluebf{\tilde v}}|^qdt\right]
    $
    for all $\bluebf v,\bluebf{\tilde v}$ such that $\bluebf{v\tilde v}=\bluebf w$.
    The second assertion then follows by applying Hölder's and Jensen's inequalities together with the first assertion, and using that $N_2(\bluebf{\tilde{v}})\leq N_2(\bluebf{w})$.
\end{proof}

We now have all the necessary tools to establish the link between the martingality of the price process $S$ and the explosion of $X$, the solution to~\eqref{eds X}. This connection is made explicit in the following theorem.

\begin{theorem}\label{th : link S X}
    For $S$ defined in~\eqref{syst1 signa EFM}, it holds that 
    \begin{equation}\label{eq : link S x}
        S \textit{ is a martingale on } [0,T] \iff \PP(\tau_\infty^U \leq T)=0,        
    \end{equation}
    where $\tau_\infty^U$ is the explosion time of the solution $U$ to~\eqref{eds U}.
\end{theorem}

\begin{proof}
    The proof follows the lines of that of~\cite[Theorem 4.2]{abijaber2025martingale}, localising $S$ with the sequence of stopping times $\tau_{n} = \inf\left\{t\in[0,T] : \int_0^t \braket{\bm{\ell}_{s}, \widehat{\mathbb{Y}}^{\lambda}_{0,s}}^2 ds=n\right\}$, $n \geq 1$, and using the fact that $\bm\ell_{\cdot}$ is $\mathcal F_0$-measurable. This yields the equivalence between the martingality of \( S \) and the non-explosion of $X$. The latter is equivalent to the non-explosion of \( U \) by Proposition~\ref{prop : well def of X}.
\end{proof}

\section{Proof of Theorem~\ref{th : crit EFM 1d}}\label{sec : efm proofs}

In this section, we provide the proof of the main result stated in Section~\ref{subsec : model efm}. The proof is largely inspired by~\cite[Sections 6--7]{abijaber2025martingale}. We first establish the sufficiency of the conditions ensuring martingality. We then prove the necessity, first in the case \( d= 0 \) and then for general \( d \in \mathbb{N} \).

As already mentioned in the proof of Proposition~\ref{prop : well def of X}, the case \( \rho = 0 \) is trivial, the price process always being a true martingale. If $N=1$, the process $U$ solves an SDE whose coefficients are globally Lipschitz and of at most linear growth. Hence, it admits a unique strong solution on $[0,T]$, and $\tau^U_\infty=\infty$ almost surely. We therefore conclude that $S$ is a martingale in both of these cases. From now on, we consider only the case \( \rho \neq 0 \) and \( N > 1 \).

\subsection{Sufficiency in Theorem~\ref{th : crit EFM 1d}}\label{subsec : proof sufficient}

We assume that $N$ is odd and greater than or equal to $3$, and that $\rho\bm{\sigma}^{\bluebf{2}^{\otimes N}}<0$. By Theorem~\ref{th : link S X}, it suffices to show that the solution $U$ of~\eqref{eds U} is almost surely non-explosive. With the notation $\tau_n^U$ introduced above, we will prove that
$\lim_{n\to\infty}\PP(\tau_n^U<T)=0.$
Since
$n^{2N}\mathbbm{1}_{\{\tau_n^U<T\}}
\leq U_{T\wedge\tau_n^U}^{2N}$,
it follows that
    $$\PP(\tau_n^U<T)\leq \frac{1}{n^{2N}}\EE[U_{T\wedge\tau_n^U}^{2N}].$$

Applying Itô's formula, integrating from $0$ to $T\wedge \tau^U_n$ and taking expectations, we obtain
    $$\EE[U_{T\wedge \tau^U_n}^{2N}]
    =\EE\bigg[\int_0^{T\wedge \tau^U_n}\bigg(\underbrace{-2N\lambda_2U_t^{2N}}_{\leq0}+2NU_t^{2N-1}\rho g(t)\braket{\bm{\sigma},{\bf D}^{\lambda}_t\bm\Xi\otimes\widehat{\mathbb{Y}}^{X,\lambda}_{0,t}}+N(2N-1)U_t^{2N-2}\bigg)dt\bigg].$$

Introducing the notation $\tilde U_t=\braket{\bluebf{2},{\bf D}_t^{\lambda}\bm\Xi\otimes\widehat{\mathbb Y}^{X,\lambda}_{0,t}}$ and $\bm{\hat\sigma}=\bm\sigma-\bm\sigma^{\bluebf{2}^{\otimes N}}\bluebf{2}^{\otimes N}$,
    we obtain
    \begin{align*}
        \EE\left[U_{T\wedge \tau_n^U}^{2N}\right]
        &\leq \EE\left[\int_0^{T\wedge \tau^U_n}\left(2NU_t^{2N-1}\rho g(t)\braket{\bm{\hat\sigma},{\bf D}^{\lambda}_t\bm\Xi\otimes\widehat{\mathbb{Y}}^{X,\lambda}_{0,t}}+\frac{2N}{N!}g(t)\rho\bm\sigma^{\bluebf{2}^{\otimes N}}U_t^{2N-1}\tilde{U}_t^N+N(2N-1)U_t^{2N-2}\right)dt\right].
    \end{align*}
    Since $\tilde U_t=U_t+e^{-\lambda_2t}\bm\Xi^{\bluebf 2}$ by definition of the tensor product,
    we obtain
        $$U_t^{2N-1}\tilde U_t^N
        =U_t^{3N-1}+\sum_{k=0}^{N-1}\binom{N}{k}U_t^{2N+k-1}(e^{-\lambda_2t}\bm\Xi^{\bluebf 2})^{N-k}.$$
    Applying Young's inequality, we get
        $$
        \EE\left[\int_0^{T\wedge\tau_n^U}|U_t^{2N-1}\tilde U_t^N-U_t^{3N-1}|dt\right]
        \leq
        N\EE\left[\int_0^{T\wedge\tau_n^U}\delta|U_t|^{3N-\frac32}\,dt\right]+C_\delta(\bm\Xi),
        $$
        $$\EE\left[\int_0^{T\wedge\tau_n^U}2NU_t^{2N-1}\rho g(t)\braket{\bm{\hat\sigma},{\bf D}_t^{\lambda}\bm\Xi\otimes\widehat{\mathbb Y}^{X,\lambda}_{0,t}}dt\right]
        \leq \EE\left[\int_0^{T\wedge\tau_n^U}\left(\delta U_t^{3N-1}+C_{\rho, N,\delta, g^*}\left|\braket{\bm{\hat\sigma},{\bf D}_t^{\lambda}\bm\Xi\otimes\widehat{\mathbb Y}^{X,\lambda}_{0,t}}\right|^{\frac{3N-1}{N}}\right)dt\right],$$
    where $\delta>0$ is arbitrary.
    By Lemma~\ref{lemme espérance majo EFM}, we obtain
        $$\EE\left[\int_0^{T\wedge\tau_n^U}\left|\braket{\bm{\hat\sigma},{\bf D}_t^{\lambda}\bm\Xi\otimes\widehat{\mathbb Y}^{X,\lambda}_{0,t}}\right|^{\frac{3N-1}{N}}dt\right]
        \leq
         C_{\varepsilon,\alpha}(\bm\Xi)\left(1+\EE\left[\int_0^{T\wedge\tau_n^U}\left(1+|U_t|^{\frac{3N-1}{N}(1+\alpha)(N-1+\varepsilon)}\right)dt\right]\right),$$
    where $\eps,\alpha >0$ are arbitrary.

    Putting everything together, we obtain
        $$\EE\left[U_{T\wedge\tau_n^U}^{2N}\right]\leq\EE\left[\int_0^{T\wedge\tau_n^U} f(U_t)\, dt\right]+C,
        $$
    with
        $f(x)=\left(\frac{2N}{N!}g_*\rho\bm\sigma^{\bluebf{2}^{\otimes N }}+(N+1)\delta\right)x^{3N-1}+C_1|x|^{\frac{3N-1}{N}(1+\alpha)(N-1+\eps)}+C_2 |x|^{2N-2} +C_3 |x|^{3N-3/2}+C_4,$
    where $C,C_1,C_2,C_3,C_4$ are positive constants depending on $\lambda$, $\rho$, $\delta$, $\alpha$, $\eps$, $\bm\sigma$, $N$, $g^*$, and $T$. Choosing $\delta\in\left(0,-\frac{2N}{(N+1)!}g_{*}\rho \bm\sigma^{\bluebf{2}^{\otimes N}}\right)$ and $\alpha,\varepsilon>0$ such that $\frac{3N-1}{N}(1+\alpha)(N-1+\varepsilon)<3N-1$ ensures that $f$ is bounded above, and hence that $\EE[U_{T\wedge\tau_n^U}^{2N}]$ is uniformly bounded in $n$, which concludes the proof.

\subsection{Necessity in Theorem~\ref{th : crit EFM 1d}: case \( d = 0 \)}\label{subsec : proof necessity 1d}
In this subsection, we show that the conditions on $\rho$, $\bm\sigma$, and $N$ in Theorem~\ref{th : crit EFM 1d} are necessary in the case \( d = 0 \). To do so, we follow the proof scheme of~\cite[Section 7.1]{abijaber2025martingale}. The following lemma is the analogue of \cite[Lemma 5.1]{abijaber2025martingale}, and its proof is omitted as it follows by the same arguments.

\begin{lemma}\label{lemma sup maj deter}
    Let $\lambda\in\RR^2_{+}$. Let $X : [0,T]\to \RR$ be a continuous semimartingale, and let $\widehat{\mathbb{X}}^{\lambda}_{0,t}$ denote the \textit{EFM} signature of $t \mapsto(t,X_t)$ over $[0,t]$. For every $\bluebf{w}\in V$, we denote by $N_2(\bluebf{w})$ the number of occurrences of \bluebf{2} in \bluebf{w}. Then, for every word $\bluebf{w}$, we have
    \begin{equation}\label{lem EFM maj déterministe}
        \sup_{t\in[0,T]}|\braket{\bluebf{w},\widehat{\mathbb{X}}^{\lambda}_{0,t}}|\leq C_{\lambda,\bluebf{w}}T^{|\bluebf{w}|-N_2(\bluebf{w})}\sup_{t\in[0,T]}|\braket{\bluebf{2},\widehat{\mathbb{X}}^{\lambda}_{0,t}}|^{N_2(\bluebf{w})}.
    \end{equation}
\end{lemma}
By Theorem~\ref{th : link S X}, it suffices to prove that $U$ explodes before \( T \) with positive probability. Recall that $U_t=\braket{\bluebf{2},\widehat{\mathbb{X}}^{\lambda}_{0,t}}$ satisfies
$$
dU_t=\left(-\lambda_2 U_t+\rho\braket{\bm\ell_t,\widehat{\mathbb{X}}^{\lambda}_{0,t}}\right)dt+dW_t,\qquad U_0=0.
$$

Assume $\rho\bm{\sigma}^{\bluebf{2}^{\otimes N}}>0$ and \( N \) is odd. Using the notation $\tau_n^U$ introduced above, we show that $\PP(\tau_\infty^U<T)>0$.

By Lemma~\ref{lemma sup maj deter}, we have 
\[
\sup_{s\in[0,t]}\left|\braket{\bm\ell_{s}-g(s)\bm{\sigma}^{\bluebf{2}^{\otimes N}}\bluebf{2}^{\otimes N}, \widehat{\mathbb{X}}^{\lambda}_{0,s}}\right|\leq C_{\bm{\sigma}, T}\sup_{s \in [0, t]}\lVert\bm\ell_{s}\rVert_1\left(1+\sup_{s\in[0,t]}|U_s|^{N-1}\right)\quad\textit{a.s.},
\]
where we denote $\lVert\bm\ell_{t}\rVert_1 := \sum_{|\bluebf{v}| \leq N}|\bm\ell_t^{\bluebf{v}}|$. For $R > 0$, define the event 
\begin{equation}\label{eq:event:B}
    \mathcal B_R=\bigg\{\sup_{s \in [0, T]}\lVert\bm\ell_{s}\rVert_1\leq R\bigg\},
\end{equation}
with $R$ large enough so that $\mathbb{P}(\mathcal{B}_R) > 0$ (recall that \( \bm\ell \) is given by \eqref{eq:ell_def}). On the event $\mathcal{B}_{R}$, we have 
\begin{equation}\label{eq : majo2}
\frac{\bm{\ell}^{\bluebf{2}^{\otimes N}}_{t}}{N!}U_t^N- C_{\bm\sigma,T, R}\left(1+\sup_{s\in[0,t]}|U_s|^{N-1}\right)\leq \braket{{\bm\ell}_{t},\widehat{\mathbb{X}}^{\lambda}_{0,t}} \leq\frac{\bm{\ell}^{\bluebf{2}^{\otimes N}}_{t}}{N!}U_t^N+C_{\bm\sigma,T,R}\left(1+\sup_{s\in[0,t]}|U_s|^{N-1}\right).
\end{equation}

For $\alpha\in\mathbb R$, define $\mathcal A_\alpha := \left\{\sup_{t \in [0, T]}\left| W_t - \alpha t \right| \leq 1\right\}$. We have $\PP(\mathcal A_\alpha)>0$. We show that, for $\alpha$ large enough, $\mathcal A_\alpha\cap\mathcal B_R\subset\{\tau_\infty^U<T\}$.

Since $\mathcal B_R\in\mathcal F_0$ and $\mathcal A_\alpha$ is independent of $\mathcal F_0$, the two events are independent and 
$$
\PP(\mathcal A_\alpha\cap\mathcal B_R)=
\PP(\mathcal A_\alpha)\PP(\mathcal B_R)>0.
$$

\textbf{Step 1: Show that $\PP(\tau_1^U<T)>0$.}
We first explain how choosing $\alpha$ large enough ensures that $\tau_1^U<T$ on the event $\mathcal{A}_\alpha\cap\mathcal{B}_R$. By~\eqref{eq : majo2}, on $\mathcal{A}_\alpha\cap\mathcal{B}_R$ and for all $t\in[0,\tau_1^U\wedge T]$,
\begin{align*}
    U_t&\geq W_t +\int_0^t \left(\frac{\rho\bm{\sigma}^{\bluebf{2}^{\otimes N}}}{N!}g(s)U_s^N-\lambda_2U_s \right)ds -|\rho|C_{\bm\sigma,T, R}\int_0^t \left(1+\sup_{u\in[0,s]} |U_u|^{N-1}\right)ds\\
    &\geq -1+(\alpha-\lambda_2-2|\rho|C_{\bm\sigma,T, R})t +\frac{\rho\bm{\sigma}^{\bluebf{2}^{\otimes N}}}{N!}\int_0^tg(s)U_s^Nds.
\end{align*}
By the comparison theorem for ODEs, $U_t\geq Y_0(t)$ for all $t\in[0,\tau_1^U\wedge T]$, where $Y_0$ is the solution to
    $$Y_0'(t)=F_0(t, Y_0(t)), \quad Y_0(0)=-1,$$
with $F_0(t, x)=\alpha-\lambda_2-2|\rho|C_{\bm\sigma,T, R}+\frac{\rho\bm{\sigma}^{\bluebf{2}^{\otimes N}}}{N!}g(t)x^N.$
The derivative at $0$ is given by $Y_0'(0)=F_0(0, -1)$.
Choose $\alpha$ large enough so that $\inf_{(t, x) \in [0, T] \times [-1, \infty)}F_0(t, x) =: \tilde{C} = \tilde C({\alpha,\lambda, \rho, \bm\sigma,T, R, g}) >0$. Hence, $F_0(t, Y_0(t)) = Y_0'(t) \geq \tilde{C} >0$ for all $t\in[0,\tau_1^U\wedge T]$, which implies that
    $$Y_0(t)\geq -1+\tilde{C}t, \quad t\in[0,\tau_1^U\wedge T].$$
If $T<\tau_1^U$, we obtain the estimate
$T\leq 2\tilde{C}^{-1}$.
Since $\tilde{C}$ tends to $+\infty$ as $\alpha$ tends to $+\infty$, we may choose $\alpha$ large enough so that this inequality fails. Hence, on the event $\mathcal{A}_\alpha\cap\mathcal{B}_R$,
\begin{equation}\label{eq tau_1}
    \tau_1^U\leq\tau_1^Y\leq T\wedge \frac{2}{\tilde{C}},
    \quad
    U_{\tau_1^U}=1.
\end{equation}

\textbf{Step 2: Show that $\PP(\tau_\infty^U<T)>0$.} We choose $\alpha$ large enough so that
\begin{align}
    &b_{k} := \alpha -\lambda_2(k+1)-|\rho|C_{\bm\sigma,T, R}(1+(k+1)^{N-1})+\frac{\rho\bm{\sigma}^{\bluebf{2}^{\otimes N}}}{N!}\bar g_{k}(k-2)^N >0, \quad k\geq 1,\label{eq necessite 1d nbres dérivées} \\
    &\frac{2}{\tilde{C}}
    +\sum_{k\geq 1}\frac{3}{b_{k}}
    <T,
    \label{eq som tau inf T}
\end{align}
where $\bar g_{k} = g^*$ if $k-2 < 0$, and $\bar g_{k} = g_{*}$ otherwise.

The existence of such an $\alpha$ is ensured by the assumption
$\rho\bm{\sigma}^{\bluebf{2}^{\otimes N}} >0$.
We now show that, on the event $\mathcal{A}_\alpha\cap\mathcal{B}_R$, $\tau_\infty^U$ is bounded by the left-hand side of~\eqref{eq som tau inf T}. More precisely, we prove that, for every $n\geq 1$,
    $$\tau_n^U-\tau_1^U\leq\sum_{1\leq k\leq n-1}\frac{3}{b_{k}},\quad U_{\tau_n^U}=n.$$
The base case $n=1$ was established in Step~1. Let $n\geq 1$ and assume that the result holds up to rank $n$. Using~\eqref{eq tau_1} and~\eqref{eq som tau inf T}, we obtain $\tau_n^U<T$. On the event $\mathcal{A}_\alpha\cap\mathcal{B}_R$, for all $t\in[\tau_n^U,\tau_{n+1}^U\wedge T]$, by~\eqref{eq : majo2},
\begin{align*}
    U_t&\geq U_{\tau_n^U}+W_t-W_{\tau_n^U}+\int_{\tau_n^U}^t \left(\frac{\rho\bm{\sigma}^{\bluebf{2}^{\otimes N}}}{N!} g(s)U_s^N-\lambda_2U_s\right) ds -|\rho|C_{\bm\sigma,T, R}\int_{\tau_{n}^U}^t \left(1+\sup_{u\in[0,s]} |U_u|^{N-1}\right)ds\\
    &\geq n-2+(\alpha-\lambda_2(n+1)-|\rho|C_{\bm\sigma,T, R}(1+(n+1)^{N-1}))(t-\tau_n^U) + \frac{\rho\bm{\sigma}^{\bluebf{2}^{\otimes N}}}{N!}\int_{\tau_n^U}^t  g(s)U_s^N ds.
\end{align*}
By the comparison theorem for ODEs, $U_t\geq Y_n(t)$ for all $t\in[\tau_n^U,\tau_{n+1}^U\wedge T]$, where $Y_n$ is the solution to
    $$Y_n'(t)=F_n(t, Y_n(t)), \quad Y_n(\tau_n^U)=n-2,$$
with $F_n(t, x)=\alpha-\lambda_2(n+1)-|\rho|C_{\bm\sigma,T, R}(1+(n+1)^{N-1})+\frac{\rho\bm{\sigma}^{\bluebf{2}^{\otimes N}}}{N!}g(t)x^N$. 

Notice that $\inf_{(t, x) \in [\tau_n^U,\tau_{n+1}^U\wedge T]\times[n-2, +\infty)}F_{n}(t, x) \geq b_{n} > 0$ by~\eqref{eq necessite 1d nbres dérivées}. An analysis similar to that in Step~1 shows that $Y_n(t)\geq n-2+b_{n}(t-\tau_n^U)$ for all $t\in[\tau_n^U,\tau_{n+1}^U\wedge T]$. 

If $T<\tau_{n+1}^U$, the above bound contradicts the induction hypothesis in view of~\eqref{eq som tau inf T}. Hence, $\tau_{n+1}^U\leq T$ and
    $$\tau_{n+1}^U-\tau_n^U
    \leq \tau_{n+1}^Y-\tau_n^U
    \leq
    \frac{3}{b_{n}},$$
with $U_{\tau_{n+1}^U}=n+1$. The desired bound on $\tau_{n+1}^U-\tau_1^U$ then follows from the induction hypothesis.

By~\eqref{eq som tau inf T} and the induction result established above, we have shown that
$\mathcal{A}_\alpha\cap\mathcal{B}_R\subset\{\tau_\infty^U<T\}$,
which concludes the proof.

The case where $N$ is even is treated similarly: on $\mathcal{A}_\alpha \cap \mathcal{B}_{R}$, with $\alpha$ of the same sign as $\rho\bm{\sigma}^{\bluebf{2}^{\otimes N}}$ and $|\alpha|$ sufficiently large, $U$ blows up before time $T$, to $+\infty$ if $\rho\bm{\sigma}^{\bluebf{2}^{\otimes N}}>0$ and to $-\infty$ if $\rho\bm{\sigma}^{\bluebf{2}^{\otimes N}}<0$.

\subsection{Necessity in Theorem~\ref{th : crit EFM 1d}: arbitrary \( d \in \mathbb{N}\)}\label{subsec : proof necessity multid}

Let us show that if $\rho\bm{\sigma}^{\bluebf{2}^{\otimes N}}>0$, then the process $U$ defined by~\eqref{eds U} blows up to $+\infty$ in finite time with positive probability. The case where $\rho\bm{\sigma}^{\bluebf{2}^{\otimes N}}<0$ and $N$ is even is treated similarly, showing that $U$ blows up to $-\infty$ with positive probability.

Our goal is to construct an event on which $U$ explodes before time $T$. 

Fix $t_0>0$, $M\in\mathbb N$, and $\varepsilon>0$ such that $N\varepsilon<1$, and define $\tau_n^U=\inf\{t\in[t_0,T]:|U_t|\geq n\}$. We further set $\tilde{\tau}_{n+1}^U=\tau_n^U+Cn^{-N}$, where $C$ is a deterministic constant, depending on $\bm\sigma$, $\rho$, $N$, $g_{*}$ and $g^*$, to be determined later. We define the event \( \mathcal{A} \) as the set on which the following inequalities hold:
\begin{align}
U_{t_0} &\geq M, \label{eq:cond1}\\
| \langle \bluebf{w},\widehat{\mathbb{Y}}^{X, \lambda}_{0, t_0} \rangle | &\leq 1,
\quad \text{for all $\bluebf{w}$, $|\bluebf{w}| \leq N$ containing at least one of the letters $\bluebf{3}, \ldots, \bluebf{d + 2}$;} \label{eq:cond2}\\
\sup_{s \in [\tau_n^U, \tau_{n+1}^{U} \land\tilde\tau_{n+1}^U]} | \langle \bluebf{w} , \widehat{\mathbb{Y}}^{X, \lambda}_{0, s} \rangle | &\leq n^{N_{\bluebf{2}}(\bluebf{w})+\epsilon |\bluebf{w}|},
\quad \text{for all $n \geq M$ and all words $\bluebf{w}$ with $|\bluebf{w}| \leq N$;} \label{eq:cond3}\\
\sup_{s \in [\tau_n^U, \tau_{n+1}^U\land\tilde\tau_{n+1}^U]} |W_s - W_{\tau_n^U}| &\leq n^{-N/2+\epsilon}, \quad \text{for all $n \geq M$.} \label{eq:cond4}
\end{align}
This event coincides with the event introduced in Section~7.2 of \cite{abijaber2025martingale}, with the classical signatures replaced by the corresponding EFM signatures and \( X \) replaced by \( U \). 

Consider the event \(\mathcal{B}_R \) defined by \eqref{eq:event:B}. We first prove that \( U \) explodes before \( T \) on $\mathcal{A} \cap \mathcal{B}_{R}$ and then that $\mathbb{P}(\mathcal{A} \cap \mathcal{B}_{R}) > 0$.

\textbf{Step 1: Show that $U$ explodes on $\mathcal{A} \cap \mathcal{B}_{R}$.} We show by induction that $C$ can be chosen such that, on $\mathcal{A} \cap \mathcal{B}_{R}$,
\begin{equation}\label{diff de tau n}
    \tau^U_{n}-\tau_{n-1}^U\leq C(n-1)^{-N},\quad U_{\tau_n^U}\geq n,\quad n \geq M,
\end{equation}
the initialisation being given by $\tau_{M-1}^U=\tau_M^U=t_0$.

Assume that $U_{\tau_n^U}\geq n$. Let us show that, on the event $\mathcal{A} \cap \mathcal{B}_{R}$,
$$\tau_{n+1}^U-\tau_n^U\leq Cn^{-N}, \quad U_{\tau_{n+1}^U}\geq n+1.$$
For all $t\in[\tau_n^U,\tilde{\tau}_{n+1}^U\wedge\tau_{n+1}^U]$, since $U$ solves~\eqref{eds U},
\begin{align*}
    U_t&=U_{\tau_n^U}+\int_{\tau_n^U}^t\left(-\lambda_2U_s+\rho \braket{\bm\ell_{s},\widehat{\mathbb{Y}}^{X,\lambda}_{0,s}}\right)\, ds+W_t-W_{\tau_n^U}\\
    &=U_{\tau_n^U}+\int_{\tau_n^U}^t\frac{\rho\bm{\sigma}^{\bluebf{2}^{\otimes N}}}{N!} g(s)U_s^Nds+\int_{\tau_n^U}^t\bigg(-\lambda_2U_s+\rho\sum_{\substack{|\bluebf{w}|\leq N\\ \bluebf{w}\neq\bluebf{2}^{\otimes N}}}\bm\ell_s^{\bluebf{w}}\braket{\bluebf{w},\widehat{\mathbb{Y}}^{X,\lambda}_{0,s}}\bigg)\,ds+W_t-W_{\tau_n^U}\\
    &\geq n-((|\rho|\sup_{s \in [0, t]}\|\bm\ell_{s}\|_{1} + \lambda_{2})Cn^{N\eps-1}+n^{-N/2+\eps})+\int_{\tau_n^U}^t \frac{\rho\bm{\sigma}^{\bluebf{2}^{\otimes N}}}{N!} g(s)U_s^N\, ds,
\end{align*}
where we used \eqref{eq:cond3}, \eqref{eq:cond4}, and $t-\tau_n^U\leq\tilde\tau_{n+1}^U-\tau_n^U=Cn^{-N}$.
For $n$ large enough, we have, on $\mathcal{A}\cap\mathcal{B}_R$,
\begin{equation}\label{eq : comparer U avec EDO}
    U_t\geq n-1 +\int_{\tau_n^U}^t \frac{\rho\bm{\sigma}^{\bluebf{2}^{\otimes N}}}{N!}g(s)U_s^N\,ds. 
\end{equation}
By the comparison theorem for ODEs,
$$Cn^{-N}\wedge(\tau_{n+1}^U-\tau_n^U)\leq \tau_{(n-1)\to(n+1)}^x,$$
where $\tau_{(n-1)\to(n+1)}^x$ denotes the time needed for the solution of the ODE
$$x'(t)=\frac{\rho\bm{\sigma}^{\bluebf{2}^{\otimes N}}}{N!}g(t)x(t)^N$$
to go from $n-1$ to $n+1$. This time is easily bounded from above by $\tilde{C}n^{-N}$, where $\tilde{C}$ is a constant depending only on $\rho$, $\bm{\sigma}$, $N$, $g_{*}$ and $g^*$.

Choosing $C>\tilde{C}$, we obtain the first statement in~\eqref{diff de tau n} at rank $n+1$.

Moreover, by the comparison theorem for ODEs,
$$U_t\geq x(t), \quad t\in[\tau_n^U,\tau_{n+1}^U],$$
where $x$ is the solution to the above ODE starting from $n-1$ at time $\tau_n^U$. Since $x(\tau_n^U)=n-1>0$ and $x'(t)>0$ whenever $x(t)>0$,
it follows that $x(t)\geq n-1$ for all $t\in[\tau_n^U,\tau_{n+1}^U]$. Hence, $U$ remains positive on $[\tau_n^U,\tau_{n+1}^U]$, and therefore $U_{\tau_{n+1}^U}\geq n+1$ by definition of $\tau_{n+1}^U$.
This concludes the proof of~\eqref{diff de tau n}.

Finally, using the first statement in~\eqref{diff de tau n}, we obtain
$\tau_\infty^U\leq t_0+\sum_{n\geq M}Cn^{-N}.$
Since $t_0$ can be chosen arbitrarily small and $M$ arbitrarily large, we may ensure that $\tau_\infty^U<T$ on $\mathcal{A} \cap \mathcal{B}_{R}$.

\textbf{Step 2: Show that $\PP(\mathcal{A} \cap \mathcal{B}_{R})> 0$.} Since $\PP(\mathcal A\cap\mathcal B_R)=\EE[\mathbbm{1}_{\mathcal B_R}\PP(\mathcal A\mid\mathcal F_0)]$ and $\PP(\mathcal B_R)>0$, it suffices to show that $\PP(\mathcal A\mid\mathcal F_0)>0$ almost surely on $\mathcal B_R$. The proof closely follows Step~2 in Section~7.2 of \cite{abijaber2025martingale}, working conditionally on $\mathcal F_0$ with constants that are uniform on $\mathcal B_R$, and is therefore omitted. We only point out the modifications required in the EFM setting. For every word $\bluebf{w}$ and every letter $\bluebf{i}$, the EFM signature coordinates satisfy
\[
    \braket{\bluebf{wi},\widehat{\mathbb{Y}}^{X,\lambda}_{0,t}}
    =e^{-\lambda_{\bluebf{wi}}(t-\tau_n^U)}
    \braket{\bluebf{wi},\widehat{\mathbb{Y}}^{X,\lambda}_{0,\tau_n^U}}
    +\int_{\tau_n^U}^t
    e^{-\lambda_{\bluebf{wi}}(t-s)}
    \braket{\bluebf{w},\widehat{\mathbb{Y}}^{X,\lambda}_{0,s}}
    \circ d \widehat{Y}_s^{i}, \quad t\in[\tau_n^U,\tau_{n+1}^U\wedge \tilde{\tau}_{n+1}^U].
    \]
Hence, the exponential inequality of~\cite[Exercise 3.16, Chapter IV]{revuzyor1999} must be adapted as follows. Its proof is a slight variation of that of the original inequality and is omitted.
\begin{lemma}\label{lemme : bernstein adapt}
    Let $W$ be a one-dimensional Brownian motion defined on a filtered probability space $(\Omega,\Ff,(\Ff_t)_{t\geq 0},\PP)$, let $\tau_1,\tau_2$ be stopping times, let $\lambda,h>0$ be deterministic constants, and let $\gamma$ be a progressively measurable process. Assume that $\tau_1\leq \tau_2\leq \tau_1+h$. Then, for all $\alpha,\eta>0$,
    \begin{equation}\label{eq : exp ineq}
        \PP\left(\sup_{t\in[\tau_1, \tau_2]}\left|\int_{\tau_1}^t e^{-\lambda(t-s)}\gamma_sdW_s\right|\geq \alpha, \sup_{s\in[\tau_1, \tau_2]}|\gamma_s|\leq \eta\right)\leq 2\exp\left(-\frac{\alpha^2}{2he^{2\lambda h}\eta^2}\right).
    \end{equation}
\end{lemma}

\bibliographystyle{plainnat}
\bibliography{references}

\end{document}